\documentclass[submission,copyright,creativecommons]{eptcs}
\providecommand{\event}{NCMA 2023} 

\usepackage{amsthm}
\usepackage{amsmath}
\usepackage{amssymb}
\usepackage{amsfonts}
\usepackage{graphicx}
\usepackage{verbatim}
\usepackage{hyperref}

\newtheorem{theorem}{Theorem}
\newtheorem{lemma}{Lemma}
\newtheorem{claim}{Claim}
\newtheorem{corollary}{Corollary}
\theoremstyle{definition}
\newtheorem{definition}{Definition}

\title{Sweeping Permutation Automata}

\author{Maria Radionova
\institute{Department of Mathematics and Computer Science \\
St.\ Petersburg State University \\
Saint Petersburg, Russia}
\email{radmarale@gmail.com}
\and
Alexander Okhotin
\institute{Department of Mathematics and Computer Science \\
St.\ Petersburg State University \\
Saint Petersburg, Russia}
\email{\quad alexander.okhotin@spbu.ru}
}
\def\titlerunning{Sweeping permutation automata}
\def\authorrunning{M. Radionova \& A. Okhotin}
\begin{document}
\maketitle

\begin{abstract}
This paper introduces sweeping permutation automata,
which move over an input string in alternating left-to-right and right-to-left sweeps
and have a bijective transition function.
It is proved that these automata recognize the same family of languages
as the classical one-way permutation automata
(Thierrin, ``Permutation automata'', \emph{Mathematical Systems Theory}, 1968).
An $n$-state two-way permutation automaton is transformed
to a one-way permutation automaton with 
$F(n)=\max_{k+\ell=n, m \leqslant \ell}
k \cdot { \ell \choose m} \cdot {k - 1 \choose \ell - m} \cdot (\ell - m)!$ states.
This number of states is proved to be necessary in the worst case,
and its growth rate is estimated as
$F(n) = n^{\frac{n}{2} - \frac{1 + \ln 2}{2}\frac{n}{\ln n}(1 + o(1))}$.
%
\end{abstract}

\sloppy

\section{Introduction}

\emph{Permutation automata}, introduced by Thierrin~\cite{Thierrin},
are one-way deterministic finite automata,
in which the transition function by each symbol forms a permutation of the set of states.
They recognize a proper subfamily of regular languages:
for instance, no finite language is recognized by any permutation automaton.
The language family recognized by permutation automata
is known as the \emph{group languages},
because their syntactic monoid is a group,
and it has received some attention
in the literature on algebraic automata theory~\cite{MargolisPin}.
Recently, Hospod\'ar and Mlyn\'ar\v{c}ik~\cite{HospodarMlynarcik}
determined the state complexity of operations on these automata,
while Rauch and Holzer~\cite{RauchHolzer}
investigated the effect of operations on permutation automata
on the number of accepting states.

Permutation automata are reversible,
in the sense that,
indeed, knowing the current state and the last read symbol
one can always reconstruct the state at the previous step.
The more general \emph{reversible automata},
studied by Angluin~\cite{Angluin} and by Pin~\cite{Pin1987},
additionally allow undefined transitions,
so that the transition function by each symbol is injective.
Reversible automata still cannot recognize all regular languages~\cite{Pin1992},
but since they can recognize all finite languages,
they are a more powerful model than permutation automata.

The notion of reversible computation has also been studied
for two-way finite automata.
In general, a two-way automaton (2DFA)
operates on a string delimited by a left end-marker ($\vdash$)
and a right end-marker ($\dashv$),
and may move its head to the left or to the right in any transition.
For the reversible subclass of two-way finite automata (2RFA),
Kondacs and Watrous~\cite{KondacsWatrous}
proved that 2RFA can recognize every regular language.
Later, Kunc and Okhotin~\cite[Sect.~8.1]{KuncOkhotin_reversible}
showed that every regular language can still be recognized
by 2RFA with no undefined transitions on symbols of the alphabet,
and with injective functions on the end-markers.
But since the latter automata,
in spite of having some kind of bijections in their transition functions,
recognize all regular languages,
they are no longer a model for the group languages.
And there seems to be no reasonable way
to have 2RFA act bijectively on both end-markers,
because in this case it would be impossible
to define both an accepting and a rejecting state.

Can permutation automata have any two-way generalization at all?
This paper gives a positive answer
by investigating \emph{sweeping permutation automata}.
This new model is a subclass of sweeping automata,
that is, two-way automata that may turn only at the end-markers.
In a sweeping automaton, there are left-moving and right-moving states,
and any transition in a right-moving state by any symbol other than an end-marker
must move the head to the right and lead to another right-moving state
(same for left-moving states).
The transitions by each symbol other than an end-marker
form two functions, one acting on the right-moving states,
and the other on the left-moving states.
In the proposed sweeping permutation automata, both functions must be bijections,
whereas the transition functions at the end-markers must be injective.
A formal definition of the new model is given in Section~\ref{section_definition}.

The main motivation for the study of sweeping permutation automata (2PerFA)
is that these automata recognize the same family of languages
as the classical one-way permutation automata (1PerFA).
This result is established in Section~\ref{section_transformation}
by showing that if the optimal transformation of two-way automata to one-way,
as defined by Kapoutsis~\cite{Kapoutsis,Kapoutsis_thesis},
is carefully applied to a 2PerFA, then it always produces a 1PerFA.

The next question studied in this paper is the number of states in a 1PerFA
needed to simulate an $n$-state 2PerFA.
The number of states used in the transformation in Section~\ref{section_transformation}
depends on the partition of $n$ states of the 2PerFA into $k$ right-moving and $\ell$ left-moving states,
and also on the number $m$ of left-moving states
in which there are no transitions by the left end-marker ($\vdash$).
The resulting 1PerFA has
$k \cdot { \ell \choose m} \cdot {k - 1 \choose \ell - m} \cdot (\ell - m)!$ states.
A matching lower bound
for each triple $(k, \ell, m)$, with $k>\ell$ and $m>0$,
is established in Section~\ref{section_lower_bound},
where it is proved that there exists a 2PerFA with $k$ right-moving states
and $\ell$ left-moving states, and with the given value of $m$,
such that every one-way deterministic automaton (1DFA)
recognizing the same language
must have at least
$k \cdot { \ell \choose m} \cdot {k - 1 \choose \ell - m} \cdot (\ell - m)!$ states.

The desired state complexity of transforming two-way permutation automata to one-way
should give the number of states in a 1PerFA that is sufficient and in the worst case necessary
to simulate every 2PerFA with $n$ states.
Note that the minimal 1DFA for a group language
is always a 1PerFA~\cite{HospodarMlynarcik},
and hence this is the same state complexity tradeoff as from 2PerFA to 1DFA.
The following function gives an upper bound on this state complexity.
\begin{equation*}
	F(n)
		=
	\max_{\substack{k+\ell=n \\ \: m \leqslant \ell}}
	k \cdot { \ell \choose m} \cdot {k - 1 \choose \ell - m} \cdot (\ell - m)!
\end{equation*}

This bound is proved to be precise in Section~\ref{section_F_n},
where it is shown that the maximum in the formula is reached for
$k = \lfloor \frac{n + 2}{2}\rfloor$
and $\ell = \lceil \frac{n - 2}{2} \rceil$,
that is, for $k>\ell$.
Since Section~\ref{section_lower_bound}
provides witness languages
for these values of $k$ and $\ell$
that require
$k \cdot { \ell \choose m} \cdot {k - 1 \choose \ell - m} \cdot (\ell - m)!$ states
in every 1PerFA,
this gives a lower bound $F(n)$.
Finally, the growth rate of the function $F(n)$ is estimated
as $F(n) = n^{\frac{n}{2} - \frac{1 + \ln 2}{2}\frac{n}{\ln n}(1 + o(1))}$
using Stirling's approximation.

An alternative, more general definition of sweeping permutation automata,
which allows acceptance both at the left end-marker and at the right end-marker,
is presented in Section~\ref{section_more_general}.
A proof that they still can be transformed to 1PerFA is sketched,
but the generalized transformation uses more states.

\section{Definition}\label{section_definition}

A one-way permutation automaton (1PerFA) is a one-way deterministic finite automaton (1DFA)
in which the transition function by every symbol is a bijection~\cite{Thierrin}.

This restriction is adapted to the more general \emph{sweeping automata}~\cite{Sipser_sweeping},
in which the set of states is divided into disjoint classes
of right-moving ($Q_+$) and left-moving ($Q_-$) states,
so that the automaton may turn only at the end-markers.
In the proposed \emph{sweeping permutation automata},
the transition function by each symbol forms one left-to-right bijection
and another right-to-left bijection.
Transitions at the end-markers are injective partial functions.

\begin{figure}[t]
	\centerline{\includegraphics[scale=0.5]{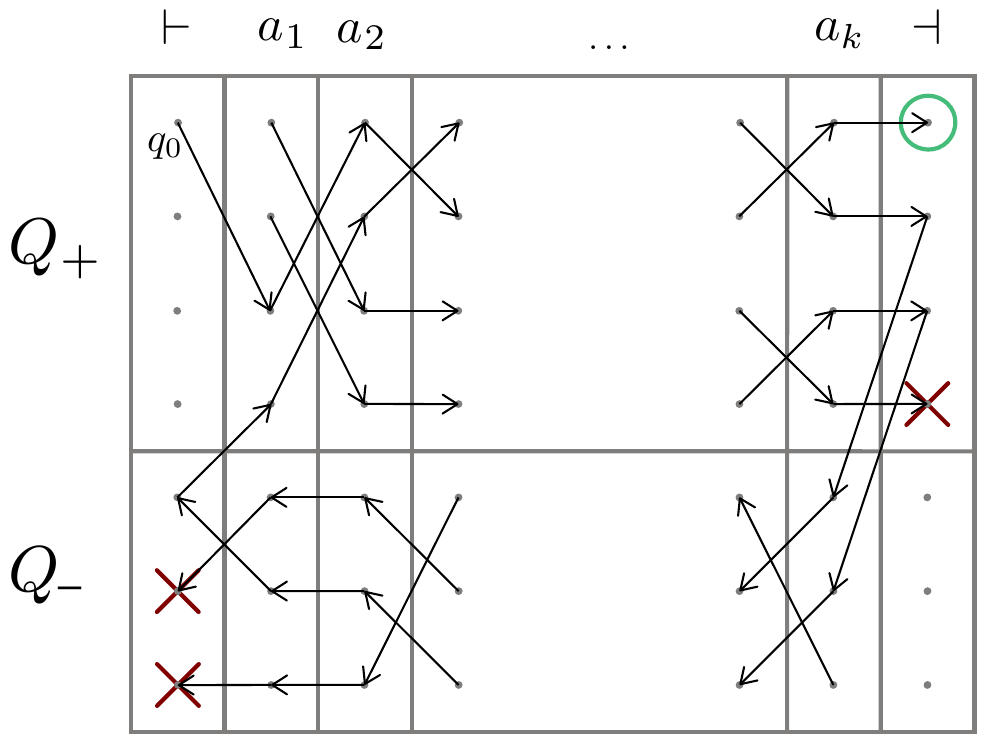}}
	\caption{Transitions of a 2PerFA on an input string.}
	\label{f:definition}
\end{figure}

\begin{definition}
A sweeping permutation automaton (2PerFA) is a 9-tuple
$\mathcal{A}=(\Sigma, Q_+, Q_-, q_0, \langle \delta_a^+ \rangle_{a \in \Sigma}, \langle \delta_a^- \rangle_{a \in \Sigma}, \delta_{\vdash}, \delta_{\dashv}, F)$, where
\begin{itemize}
\item
	$\Sigma$ is the alphabet;
\item
	$Q_+ \cup Q_-$ is the set of states, where $Q_+ \cap Q_- = \varnothing$;
\item
	$q_0 \in Q_+$ is the initial state;
\item
	for each symbol $a \in \Sigma$, $\delta_a^+ \colon Q_+ \to Q_+$ and $\delta_a^- \colon Q_- \to Q_-$ are bijective transition functions;
\item
	the transition functions at the end-markers
	$\delta_{\vdash} \colon (Q_- \cup \{q_0\}) \to Q_+$,
	$\delta_{\dashv} \colon Q_+ \to Q_-$
	are partially defined and injective on their respective domains;
\item
	$F \subseteq Q_+$ is the set of accepting states,
	with $\delta_{\dashv}(q)$ undefined for all $q \in F$.
\end{itemize}
The computation of the automaton is defined in the same way
as for sweeping automata of the general form.
Given an input string $w = a_1 \ldots a_m \in \Sigma^*$,
the automaton operates on a tape ${\vdash} a_1 \ldots a_m {\dashv}$.
Its computation is a uniquely defined sequence of configurations,
which are pairs $(q, i)$
of a current state $q \in Q_+ \cup Q_-$ and a position $i \in \{0, 1, \ldots, m+1\}$ on the tape.
It starts in the configuration $(q_0, 0)$
and makes a transition to $(\delta_{\vdash}(q_0), 1)$.
If the automaton is in a configuration $(q, i)$ with $q \in Q_+$ and $i \in \{1, \ldots, m\}$,
it moves to the next configuration $(\delta_{a_i}^+(q), i+1)$.
Once the automaton is in a configuration $(q, m+1)$,
it accepts if $q \in F$,
or moves to $(\delta_{\dashv}(q), m)$ if $\delta_{\dashv}(q)$ is defined,
and rejects otherwise.
In a configuration $(q, i)$ with $q \in Q_-$ and $i \in \{1, \ldots, m\}$,
the automaton moves to $(\delta_{a_i}^-(q), i-1)$.
Finally, in a configuration $(q, 0)$ with $q \in Q_-$,
the automaton turns back to $(\delta_{\vdash}(q), 1)$
or rejects if this transition is undefined.

The language recognized by an automaton $\mathcal{A}$,
denoted by $L(\mathcal{A})$, is the set of all strings it accepts.
\end{definition}

A one-way permutation automaton (1PerFA) is a 2PerFA in which $Q_- = \varnothing$ and $\delta_{\dashv}$ is undefined on every state. The left end-marker can be removed, making $\delta_{\vdash}(q_0)$ the new initial state.

Note that a 2PerFA never loops.
If it did, then some configuration would appear twice in some computation.
Consider the earliest such configuration.
If it is not the initial configuration,
then there exists only one possible previous configuration.
It appears at least twice in the computation,
and it precedes the configuration considered before, a contradiction.
The repeated configuration cannot be the initial configuration,
in which the 2PerFA is at the left end-marker in the state $q_0 \in Q_+$,
because the automaton may return to $\vdash$ only in the states from $Q_-$.

\section{Transformation to one-way}\label{section_transformation}

Since a 2PerFA is a 2DFA,
the well-known transformation to a one-way automaton can be applied to it~\cite{Kapoutsis, Shepherdson}:
after reading a prefix of a string $u$,
the 1DFA stores the first state
in which the 2PerFA eventually goes right from the last symbol of the prefix,
and the function which encodes the outcomes of all computations
starting at the last symbol of the prefix
and ending with the transition from that symbol to the right.
For a sweeping automaton, all computations encoded by the functions
start in $Q_-$ and end in $Q_+$.
Moreover, computations starting in different states should end in different states.
Therefore, a one-way automaton has to remember fewer different functions of a simpler form,
and eventually turns out to be a permutation automaton.

\begin{lemma}\label{2PerFA_to_1PerFA_lemma}
	For every 2PerFA $\mathcal{A} = (\Sigma, Q_+, Q_-, q_0, \langle \delta_a^+ \rangle_{a \in \Sigma}, \langle \delta_a^- \rangle_{a \in \Sigma}, \delta_{\vdash}, \delta_{\dashv}, F)$ with
	\begin{equation*}
		|Q_+| = k, \quad |Q_-| = \ell, \quad |Q_-^{\times}| = m,
	\end{equation*}
	where $Q_-^{\times} \subseteq Q_-$~is the set of states from which there is no transition by $\vdash$, there exists a 1PerFA recognizing the same language which uses states of the form $(q, f)$ satisfying the following restrictions:
	\begin{itemize}
		\item $q \in Q_+$,
		\item $f \colon Q_- \to Q_+$ is a partially defined function,
		\item $q \notin \text{ Im} f$,
		\item $f$ is injective,
		\item $f$ is undefined on exactly $m$ states.
	\end{itemize}
\end{lemma}
\begin{proof}
We will construct a 1PerFA $\mathcal{B} = ( \Sigma, Q, \widetilde{q}_0, \widetilde{\delta}, \widetilde{F} )$;
states of $\mathcal{B}$ shall be pairs $(q, f)$,
where $q \in Q_+$ and $f \colon Q_- \to Q_+$ is a partial function.

After reading a prefix $s \in \Sigma^*$,
the automaton $\mathcal{B}$ should come to a state $(q, f)$,
where $q$ and $f$ describe the outcomes of the following computations of $\mathcal{A}$ on $s$:
\begin{itemize}
\item
	if the 2PerFA starts on ${\vdash}s$ in its initial configuration,
	then it eventually moves from the last symbol of ${\vdash}s$ to the right in the state $q$,
\item
	if the 2PerFA starts at the last symbol of ${\vdash}s$ in a state $p$ from $Q_-$,
	then it eventually leaves $s$ in the state $f(p) \in Q_+$.
	If the computation reaches an undefined transition at $\vdash$, then the value $f(p)$ is undefined.
\end{itemize}

The initial state is defined as
\begin{equation*}
\widetilde{q_0} = (\delta_{\vdash}(q_0), \delta_{\vdash} \big|_{Q_-})
\end{equation*}
where $\delta_{\vdash} \big|_{Q_-}$ is $\delta_{\vdash}$ restricted to the domain $Q_-$.

The definition of the transition function is as follows:
\begin{align*}
	\widetilde{\delta}_a((q, f))
		=
	(\delta_a^+(q), \delta_a^+ \circ f \circ \delta_a^-),
	&& \text{for all } a \in \Sigma.
\end{align*}

\begin{claim}\label{reachable_states_claim}
Every pair $(q, f)$ reachable from $\widetilde{q}_0$ by transitions in $\widetilde{\delta}$ satisfies the following conditions:
\begin{itemize}
	\item $q \in Q_+$,
	\item $f \colon Q_- \to Q_+$ is a partially defined function,
	\item $q \notin \text{ Im} f$,
	\item $f$ is injective,
	\item $f$ is undefined on exactly $m$ states.
\end{itemize}
\end{claim}
\begin{proof}
	Induction on the length of the string.
	Let $(q, f)$ be reachable in $\mathcal{B}$ by a string $u$. If $u = \varepsilon$, then $(q, f)$ is the initial state
	\begin{equation*}
		(q, f) = (\delta_{\vdash}(q_0), \delta_{\vdash} \big|_{Q_-})
	\end{equation*}
	The first two conditions are satisfied
	because the domain of $\delta_{\vdash}$ is split into $q_0$ and $Q_-$.
	The third and the fourth conditions follow from the injectivity of $\delta_{\vdash}$
	and the disjointness of $\{ q_0 \}$ and $Q_-$.
	The states on which $\delta_{\vdash} \big|_{Q_-}$
	is not defined are the states in $Q_-^{\times}$ by definition, and there are $m$ of them.
	
	Let $(q, f)$ be reachable in $\mathcal{B}$ by a string $u$ and let $(q', f')$ be reachable from it by a transition by $a$. The induction assumption is true for the state $(q, f)$, and $(q', f')$ is defined as
	\begin{equation*}
		(q', f') = (\delta_a^+(q), \delta_a^+ \circ f \circ \delta_a^-)
	\end{equation*}
	The state $q' \in Q_+$ because $\delta_a^+$ is a total function which acts from $Q_+$ to $Q_+$. The function $f'$ acts from $Q_-$ to $Q_+$ because $\delta_a^-$ acts from $Q_-$ and $\delta_a^+$ acts to $Q_+$. To see that $\delta_a^+(q) \notin \text{Im } \delta_a^+ \circ f \circ \delta_a^-$, consider that $\{ q\}$ and $\text{Im } f$ are disjoint by the induction assumption, and therefore their images under a bijection $\delta_a^+$, that are, $\{ \delta_a^+(q) \}$ and $\text{Im } \delta_a^+ \circ f$, are disjoint as well.
	
	The function $\delta_a^+ \circ f \circ \delta_a^-$ is injective as a composition of injective functions. The function $\delta_a^+ \circ f \circ \delta_a^-$ is undefined on exactly $m$ states because $f$ is, and functions $\delta_a^+$ and $\delta_a^-$ are total bijections.
\end{proof}
Let $Q$ be the set of all pairs $(q, f)$ satisfying Claim~\ref{reachable_states_claim}.

\begin{claim}
After reading a prefix $s \in \Sigma^*$ the automaton $\mathcal{B}$ comes to a state $(q, f)$, where
\begin{itemize}
\item
	if the 2PerFA starts on ${\vdash}s$ in its initial configuration,
	then it eventually moves from the last symbol of ${\vdash}s$ to the right in the state $q$,
\item
	if the 2PerFA starts at the last symbol of ${\vdash}s$ in a state $p$ from $Q_-$,
	then it eventually leaves $s$ in the state $f(p) \in Q_+$.
	If the computation reaches an undefined transition at $\vdash$, then the value $f(p)$ is undefined.
\end{itemize}
\end{claim}
\begin{proof}
Induction on the length of the string.
It is clear for the empty string and the initial state.
Let $(q, f)$ be the state of $\mathcal{B}$ after reading $s$,
then $(q', f') = \widetilde{\delta}_a((q, f))$ is the state after reading $sa$.
By the induction assumption, the state $(q, f)$ and the string $s$ satisfy the property.
Then $\mathcal{B}$ reads the symbol $a$ and comes to the state
$(\delta_a^+(q), \delta_a^+ \circ f \circ \delta_a^-)$.
The automaton eventually leaves ${\vdash}s$ to the right in the state $q$;
then it comes to $a$ in this state and makes a transition to $\delta_a^+(q)$,
thus leaving ${\vdash}sa$ to the right.
To prove the second condition, let the 2PerFA start on ${\vdash}sa$
at the symbol $a$ in a state $p \in Q_-$.
Then it moves left to the last symbol of ${\vdash}s$ in the state $\delta_a^-(p)$.
Then the computation continues on the string ${\vdash}s$ and its outcome is given by the function $f$.
Eventually the 2PerFA leaves ${\vdash}s$ to the right
and comes to $a$ in the state $f(\delta_a^-(p))$.
Then the 2PerFA looks at the symbol and goes to $\delta_a^+(f(\delta_a^-(p)))$ moving to the right.
If $f(\delta_a^-(p))$ is undefined, then so is $f'(p)$.
So, the function $\delta_a^+ \circ f \circ \delta_a^-$ indeed satisfies the second claim.
\end{proof}

\begin{figure}[t]
	\centerline{\includegraphics[scale=0.5]{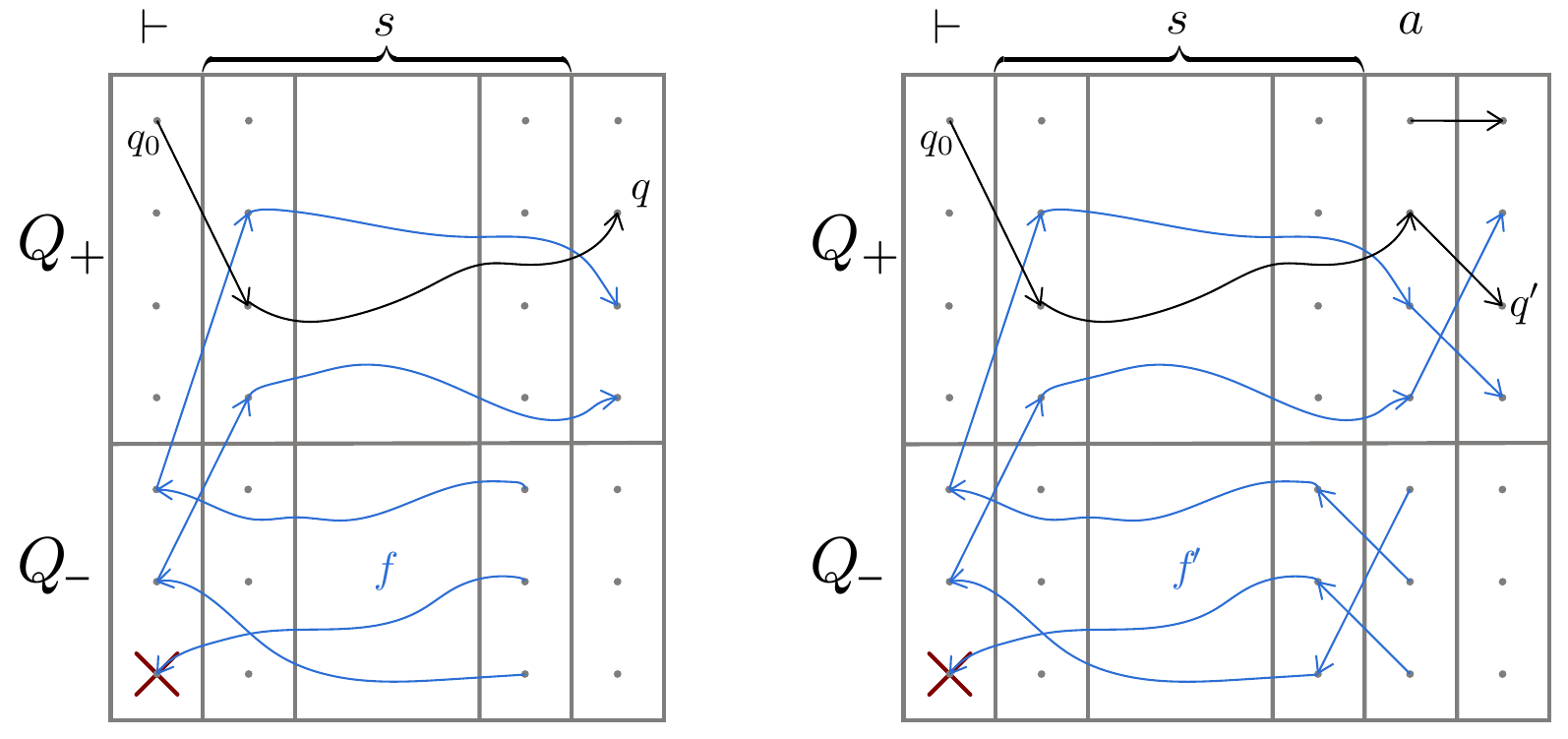}}
	\caption{(left) $\mathcal{B}$ in a state $(q, f)$;
		(right) transition of $\mathcal{B}$ by $a$.}
	\label{f:transition_building}
\end{figure}
To define $(q, f)$ as an accepting or a rejecting state, consider the following sequence of states $\{q_i\}_{i \geqslant 1}$ with $q_i \in Q_+$. The first element is
$$q_1 = q$$
For each $q_i$ if the 2PerFA has a transition by $\dashv$ from $q_i$ and the function $f$ is defined on $\delta_{\dashv}(q_i)$ then
\begin{equation*}
	q_{i + 1} = f(\delta_{\dashv}(q_i))
\end{equation*}
Otherwise the sequence ends. The sequence $\{ q_i \}_{i \geqslant 1}$ is always finite, because if it loops then some state $\widetilde{q}$ appears at least twice. Consider the earliest repeated state. If it is not $q_1$ then there is the previous one. The previous state for $\widetilde{q}$ is the same for all its appearances as $f(\delta_{\dashv})$ is an injective function. Therefore, $\widetilde{q}$ is not the earliest repeated state. So, $\widetilde{q}$ should be $q_1$. As $q_1 \notin \text{Im }f$ the state $q_1$ cannot be repeated, a contradiction.

If this sequence ends with an accepting state $q_i \in F$, then the state $(q, f)$ is accepting in $\mathcal{B}$. Otherwise, the state $(q, f)$ is rejecting.

The constructed one-way automaton accepts the same language as the 2PerFA because when it comes by some string $s$ to a state $(q, f)$, then before accepting or rejecting on ${\vdash}s{\dashv}$ the 2PerFA  passes through the sequence of states $\{ q_i \}_{i \geqslant 1}$, with $q_i \in Q_+$, at $\dashv$.

\begin{claim}
	The resulting one-way automaton is a permutation automaton.
\end{claim}
\begin{proof}
We will prove that the transition function $\widetilde{\delta}_a$
is a bijection for each symbol $a \in \Sigma$. Firstly, we show its injectivity. Let
\begin{equation*}
	\widetilde{\delta}_a((q_1, f_1)) = \widetilde{\delta}_a((q_2, f_2))
\end{equation*}
Then, by the definition of $\widetilde{\delta}_a$,
\begin{equation*}
	(\delta_a^+(q_1), \delta_a^+ \circ f_1 \circ \delta_a^-) = (\delta_a^+(q_2), \delta_a^+ \circ f_2 \circ \delta_a^-)
\end{equation*}
Then $\delta_a^+(q_1) = \delta_a^+(q_2)$, which means that $q_1 = q_2$, because $\delta_a^+$ is a bijection. Then consider the equality
\begin{equation*}
	\delta_a^+ \circ f_1 \circ \delta_a^- = \delta_a^+ \circ f_2 \circ \delta_a^-
\end{equation*}
Taking a composition of both sides of the equation
with $(\delta_a^+)^{-1}$ on the left and $(\delta_a^-)^{-1}$ on the right yields
\begin{equation*}
	(\delta_a^+)^{-1} \circ \delta_a^+ \circ f_1 \circ \delta_a^- \circ (\delta_a^-)^{-1} = (\delta_a^+)^{-1} \circ \delta_a^+ \circ f_2 \circ \delta_a^- \circ (\delta_a^-)^{-1}
\end{equation*}
Then $f_1 = f_2$, and the injectivity is proved. The function $\widetilde{\delta_a}$ is total and has equal domain and range, it is therefore a bijection.
\end{proof}

This completes the proof of the lemma.
\end{proof}

\begin{theorem}\label{upper_bound_theorem}
For every 2PerFA $\mathcal{A} = (\Sigma, Q_+, Q_-, q_0, \langle \delta_a^+ \rangle_{a \in \Sigma}, \langle \delta_a^- \rangle_{a \in \Sigma}, \delta_{\vdash}, \delta_{\dashv}, F)$ with
\begin{equation*}
	|Q_+| = k, \quad |Q_-| = \ell, \quad |Q_-^{\times}| = m,
\end{equation*}
where $Q_-^{\times} \subseteq Q_-$~is the set of states from which there is no transition by $\vdash$,
there exists a 1PerFA with at most
\begin{equation*}
k \cdot { \ell \choose m} \cdot {k - 1 \choose \ell - m} \cdot (\ell - m)!
\end{equation*}
states that recognizes the same language.
\end{theorem}
\begin{proof}
Consider the one-way automaton $\mathcal{B}$ obtained for the 2PerFA $\mathcal{A}$ in Lemma~\ref{2PerFA_to_1PerFA_lemma}. Every state $(q, f)$ of $\mathcal{B}$ satisfies the following conditions:
\begin{itemize}
	\item $q \in Q_+$,
	\item $f \colon Q_- \to Q_+$ is a partially defined function,
	\item $q \notin \text{ Im} f$,
	\item $f$ is injective,
	\item and $f$ is undefined on exactly $m$ states.
\end{itemize}

For a fixed $q \in Q_+$, let us count the number of functions satisfying the conditions above: firstly, we should choose $m$ states from $Q_-$ on which $f$ is not defined.
Secondly, from the $k - 1$ states we should choose $\ell - m$ different values for $f$'s range. And lastly, we can choose a bijection between these two sets in $(\ell - m)!$ ways.
\begin{equation*}
	{ \ell \choose m} \cdot {k - 1 \choose \ell - m} \cdot (\ell - m)!
\end{equation*}
By multiplying this number by $k$ (the number of different states $q$) we will get the claimed number of states.
\end{proof}

If a two-way automaton has $n$ states in total,
then there is only a finite number of partitions into left-moving and right-moving states,
and finitely many choices of $m$,
and hence the following number of states
is sufficient to transform this automaton to one-way.
\begin{equation*}
	F(n) = \max_{\substack{k, \ell, m\\ k > 0,\; \ell \geqslant m \geqslant 0\\m \geqslant \ell - k + 1}} k \cdot { \ell \choose m} \cdot {k - 1 \choose \ell - m} \cdot (\ell - m)!
\end{equation*}

\begin{corollary}\label{upper_bound_corollary}
For every $n$-state 2PerFA there exists a 1PerFA with $F(n)$ states that recognizes the same language.
\end{corollary}

Later it will be proved that $F(n)$ is a sharp bound, that is,
for some $n$-state 2PerFA every 1PerFA recognizing the same language has to have at least $F(n)$ states.

\section{Lower bound on the number of states}\label{section_lower_bound}

In this section it will be shown
that the upper bound on the number of states in a 1DFA
needed to simulate a 2PerFA
is sharp for each triple $(k, \ell, m)$, where $k > \ell$ and $\ell \geqslant m > 0$.
Only the case of $k > \ell$ is considered,
because, as it will be shown later,
the maximum over $(k, \ell, m)$ in $F(n)$ is reached for $k > \ell$
(in other words, a 2PerFA that requires the maximum number of states in a 1PerFA has $|Q_+| > |Q_-|$).

\begin{theorem}\label{lower_bound_theorem}
For all $k, \ell, m$ with $k > \ell > 0$ and $\ell \geqslant m > 0$ there exists a 2PerFA $\mathcal{A} = (\Sigma, Q_+, Q_-, q_0, \langle \delta_a^+ \rangle_{a \in \Sigma}, \langle \delta_a^- \rangle_{a \in \Sigma}, \delta_{\vdash}, \delta_{\dashv}, F)$ such that
\begin{equation*}
	|Q_+| = k, \quad |Q_-| = \ell,
\end{equation*}
the function $\delta_{\vdash}$ is undefined on exactly $m$ arguments from $Q_-$, and every 1DFA recognizing $L(\mathcal{A})$ must have at least
\begin{equation*}
	k \cdot { \ell \choose m} \cdot {k - 1 \choose \ell - m} \cdot (\ell - m)!
\end{equation*}
states.
\end{theorem}
\begin{proof}
Fix $k, \ell, m$ and consider a 2PerFA $\mathcal{A} = (\Sigma, Q_+, Q_-, q_0, \langle \delta_a^+ \rangle_{a \in \Sigma}, \langle \delta_a^- \rangle_{a \in \Sigma}, \delta_{\vdash}, \delta_{\dashv}, F)$ where
\begin{itemize}
\item
	$Q_+ = \{ q_0, \ldots, q_{k - 1} \}, Q_- = \{r_0, \ldots, r_{\ell - 1}\}$.
\item
	The initial state is $q_0$ and the accepting states are $\{ q_{\ell}, \ldots, q_{k - 1} \}$.
\item
	The functions $\delta_a^+$ and $\delta_b^+$
	are generators of the permutation group on the set $Q_+$
	(for example, these could be a cycle on all elements of $Q_+$ and an elementary transposition).
	Similarly, $\delta_c^-, \delta_d^-$ are generators of the permutation group on the set $Q_-$,
	and $\delta_a^-, \delta_b^-, \delta_c^+, \delta_d^+$ are identity functions.
\item
	Transitions at the left end-marker are $\delta_{\vdash}(q_0) = q_0$
	and $\delta_{\vdash}(r_i) = q_{i + 1}$ for $0 \leqslant i < \ell - m$.
	There are no transitions by $\vdash$ in the remaining $m$ states.
\item
	Transitions at the right end-marker are $\delta_{\dashv}(q_i) = r_i$ for $0 \leqslant i < \ell$.
	There are no transitions by $\dashv$ in the remaining $k - \ell$ states.
\end{itemize}
\begin{figure}[t]
	\centerline{\includegraphics[scale=0.5]{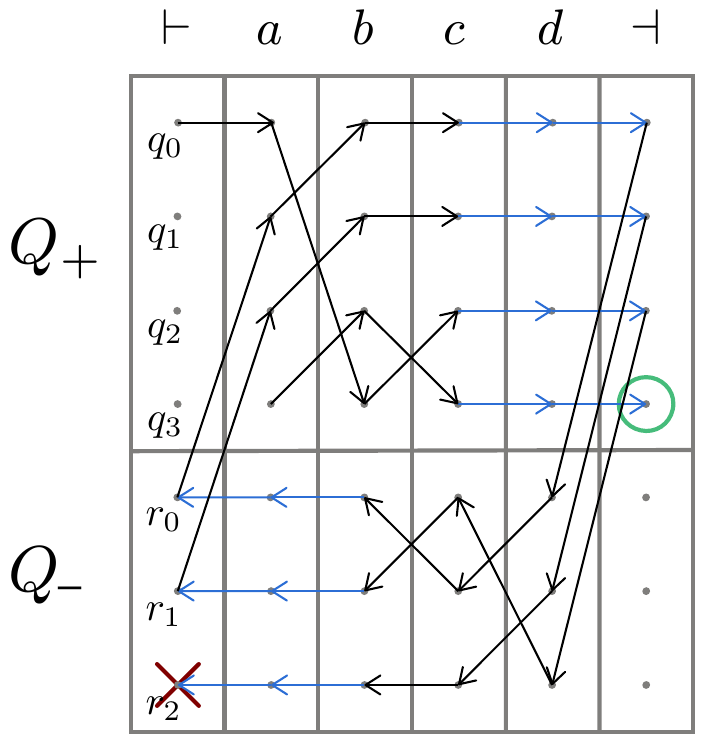}}
	\caption{Symbols $a,b,c,d$ in the construction of $\mathcal{A}$.}
	\label{f:A_building}
\end{figure}
The proof of the lower bound on the size of any 1DFA recognizing this language
is by showing that the automaton
$\mathcal{B} = (Q, \widetilde{q}_0, \widetilde{F}, \widetilde{\delta}, \Sigma)$
obtained from $\mathcal{A}$ by the transformation in Lemma~\ref{2PerFA_to_1PerFA_lemma}
will be minimal.
We will show that every state is reachable and for every two states there exists a separating string.
Firstly prove the reachability. Consider a state $(q, f)$.
It satisfies the following conditions from Lemma~\ref{2PerFA_to_1PerFA_lemma}:
\begin{itemize}
\item $q \in Q_+$,
\item $f \colon Q_- \to Q_+$ is a partially defined function,
\item $q \notin \text{Im} f$,
\item $f$ is injective,
\item $f$ is undefined on exactly $m$ states.
\end{itemize}

\begin{figure}[t]
	\centerline{\includegraphics[scale=0.5]{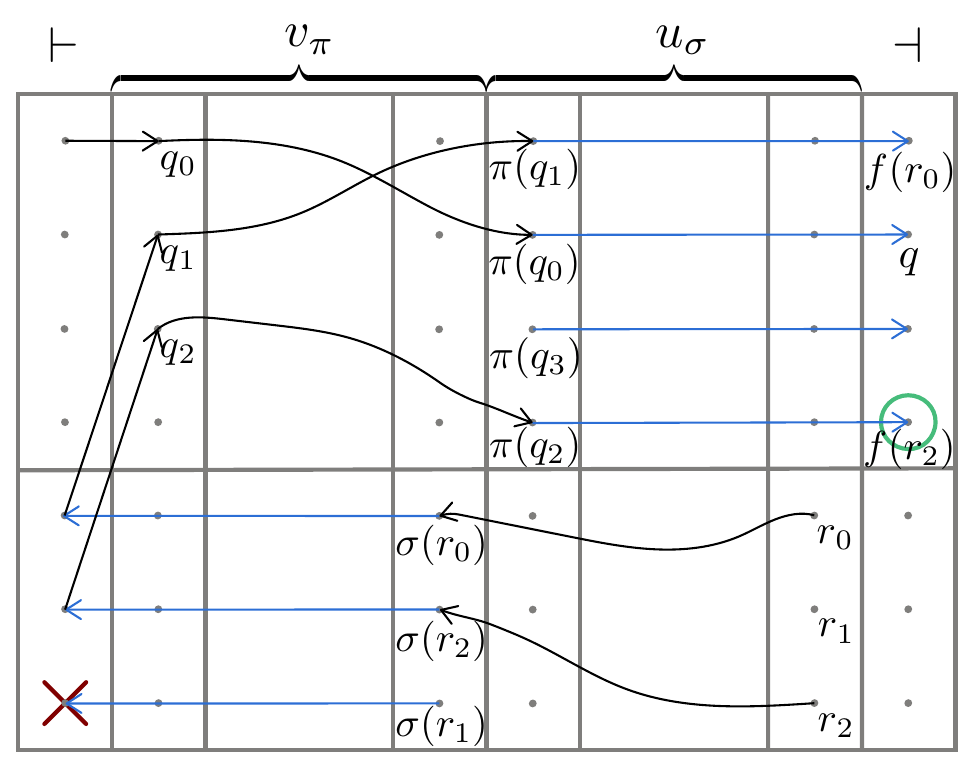}}
	\caption{The action of $\sigma$ and $\pi$.}
	\label{f:reachability}
\end{figure}

Let $\sigma$ be a permutation on the set $Q_-$
that maps the states on which $f$ is not defined
to the $m$ states from $Q_-$ without a transition by $\vdash$.
Take a string $u_{\sigma} \in \{ c, d \}^*$ that,
when read from right to left, implements the permutation $\sigma$,
and acts as an identity on $Q_+$ when read from left to right.
Next, the goal is to define a permutation $\pi$ on the set $Q_+$
that maps $q_0$ to $q$, and, for each state $q' \in Q_-$ on which $f$ is defined,
it should map the state $\delta_{\vdash}(\sigma(q'))$ to the state $f(q')$,
as shown in Figure~\ref{f:reachability}.
Note that $\delta_{\vdash}(\sigma(q'))$ is defined for all $q'$ in the domain of $f$.
We can introduce such a permutation
because each state $\delta_{\vdash}(\sigma(q'))$ is not equal to $q_0$
and they are all pairwise distinct (as $\sigma$ is a permutation and $\delta_{\vdash}$ is an injection).
Also each state $f(q')$ is not equal to $q$ and they are all pairwise distinct too.
Take the string $v_{\pi} \in \{ a, b \}^*$ that, when read from left to right, implements $\pi$,
and is an identity on $Q_-$ if read from right to left.
So, by the string $v_{\pi}u_{\sigma}$ we reach the state $(q, f)$.

Next, the existence of a separating string for all pairs of states is proved. Consider different states $(q_1, f_1)$ and $(q_2, f_2)$. Let them be reached by strings $s_1$ and $s_2$, respectively. There are several cases.
\begin{itemize}
	\item The states $q_1, q_2$ are different, as shown in Figure~\ref{f:lower_estimation_1}. Fix a state $r \in Q_-$ with no transition on the left end-marker defined. Since $\mathcal{A}$ is a permutation automaton, there exists a state $\widetilde{r} \in Q_-$ such that after reading $s_1$ from right to left starting in the state $\widetilde{r}$, the automaton is in the state $r$. Also let $q_1' \in Q_+$ be the state from which there is a transition to $\widetilde{r}$ by the right end-marker: such a state exists, because in the 2PerFA there are transitions to all states in $Q_-$ by the right end-marker. Then let $\pi$ be such a permutation on the set $Q_+$ that maps $q_1$ to $q_1'$ and $q_2$ to an accepting state $q_2' \in Q_+$. Let $v_{\pi} \in \{ a, b \}^*$ be the string that implements the permutation $\pi$ when read from left to right. So, the string $s_1v_{\pi}$ will be rejected by 2PerFA and the string $s_2v_{\pi}$ will be accepted. That is, $v_{\pi}$ is a separating string.
	\begin{figure}[t]
		\centerline{\includegraphics[scale=0.5]{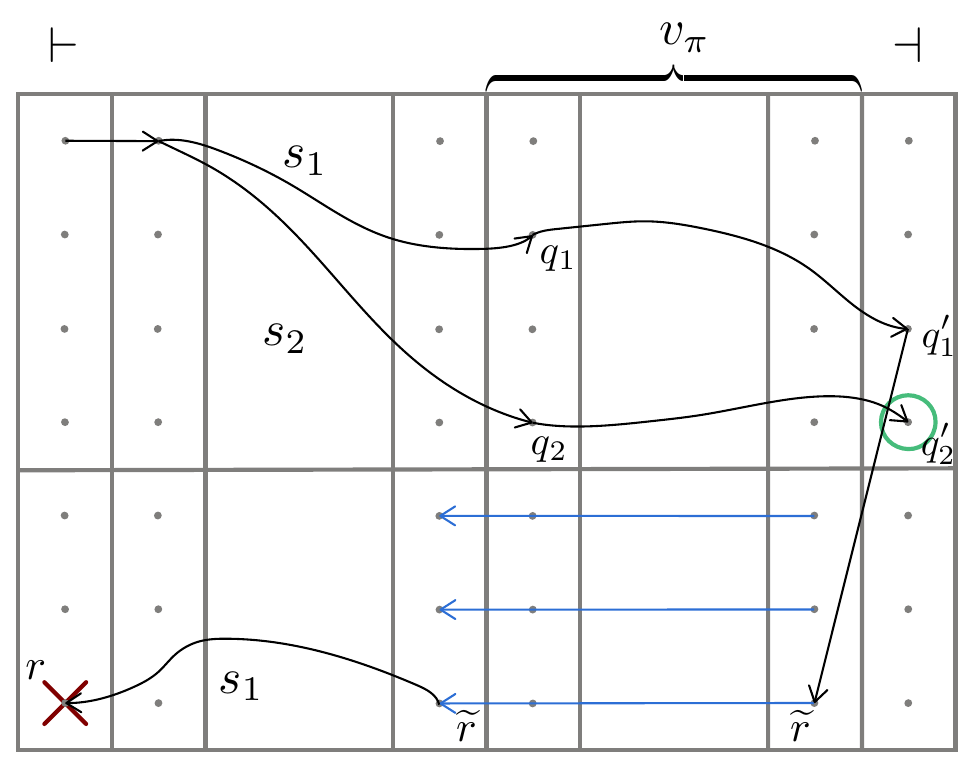}}
		\caption{A separating string for states $(q_1, f_1)$ and $(q_2, f_2)$, with $q_1 \neq q_2$.}
		\label{f:lower_estimation_1}
	\end{figure}
	\item The states $q_1, q_2$ are the same, but $f_1 \neq f_2$. Because $f_1 \neq f_2$ there exists a state $r \in Q_-$ on which these functions differ, that is, either one of $f_1(r), f_2(r)$ is defined and the other is not, or both are defined and are different states.
	\begin{enumerate}
	\item
		First, assume that $f_1$ is defined on $r$,
		but $f_2$ is not (the case of $f_2(r)$ defined and $f_1(r)$ undefined is symmetric). It is true that $f_1(r) \neq q_1$, because $q_1 \notin \text{ Im}f_1$ in every state of $\mathcal{B}$. Fix the state $q$ with a transition from it to $r$ by the right end-marker. Then let $\pi$ be such a permutation of the set $Q_+$ that maps $q_1$ to $q$ and $f_1(r)$ to an accepting state. And let the string $v_{\pi} \in \{ a, b \}^*$ implement $\pi$ when the 2PerFA reads it from left to right as shown in Figure~\ref{f:lower_estimation_2_1}. The string $s_1v_{\pi}$ will be accepted by the 2PerFA and the string $s_2v_{\pi}$ will be rejected. So, $v_{\pi}$ is a separating string.
		\begin{figure}[t]
			\centerline{\includegraphics[scale=0.5]{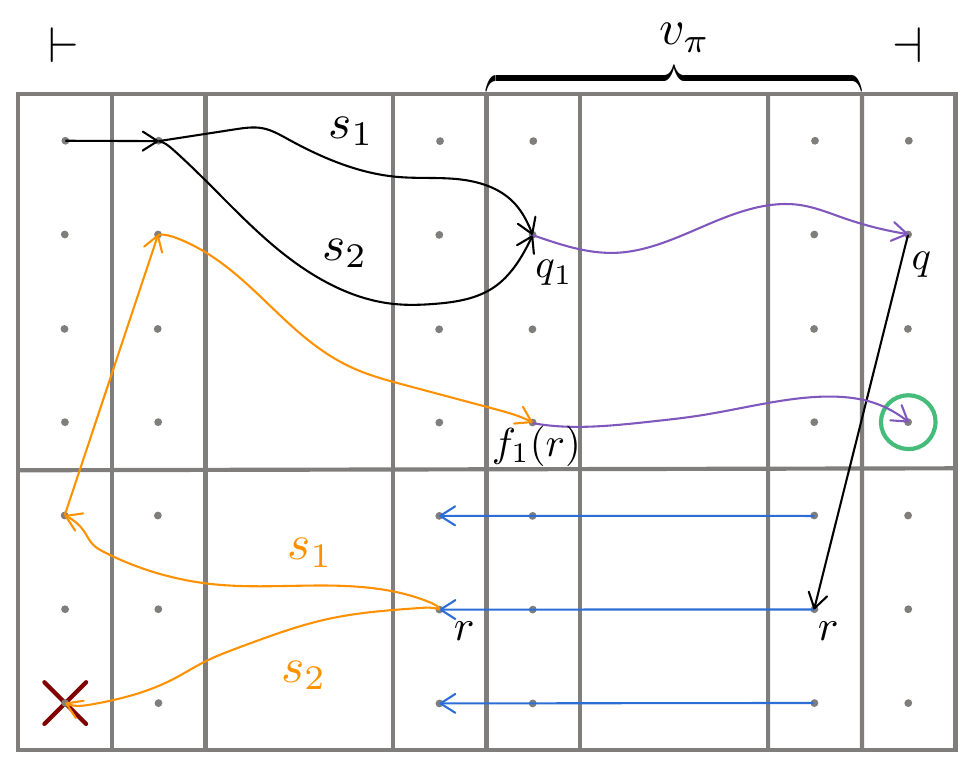}}
			\caption{A separating string for states $(q_1, f_1)$ and $(q_2, f_2)$, with $q_1 = q_2$, $f_1(r)$ defined, $f_2(r)$ undefined.}
			\label{f:lower_estimation_2_1}
		\end{figure}
	\item
		Both functions $f_1$ and $f_2$ are defined on $r$.
		It is true that $f_1(r) \neq q_1$ and $f_2(r) \neq q_1$.
		Again, fix the state $q$ from which there is a transition to $r$ by the right end-marker.
		Then fix a state $\widetilde{r}$ from $Q_-$
		without a transition by the left end-marker
		(it exists because $m \geqslant 1$).
		Let $r^* \in Q_-$ be the state from which the 2PerFA reads $s_1$
		and finishes in $\widetilde{r}$.
		We can choose such a state because the 2PerFA is a permutation automaton,
		and this state cannot coincide with $r$
		because in this case $f_1(r)$ would be undefined.
		Then let $q^* \in Q_+$ be the state,
		from which there is a transition by the right end-marker to $r^*$:
		it exists because there are such transitions to all states in $Q_-$.
		Let $v_{\pi} \in \{ a, b \}^*$ implement a permutation $\pi$ on $Q_+$
		that maps $q_1$ to $q$, $f_1(r)$ to $q^*$ and $f_2(r)$ to an accepting state,
		as illustrated in Figure~\ref{f:lower_estimation_2_2}.
		The string $s_1v_{\pi}$ will be rejected by the 2PerFA,
		and the string $s_2v_{\pi}$ will be accepted, so, $v_{\pi}$ is a separating string. 
		\begin{figure}[t]
			\centerline{\includegraphics[scale=0.5]{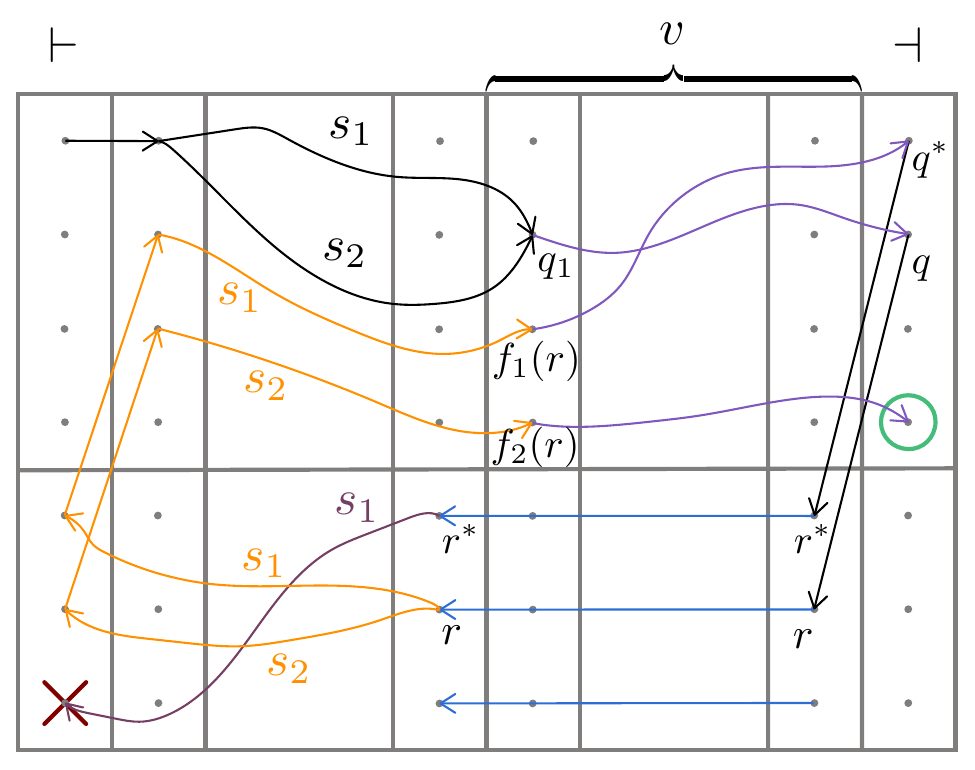}}
			\caption{A separating string for states $(q_1, f_1)$ and $(q_2, f_2)$, with $q_1 = q_2$, $f_1(r) \neq f_2(r)$.}
			\label{f:lower_estimation_2_2}
		\end{figure}
	\end{enumerate}
\end{itemize}
\end{proof}

\section{Optimal partition of $n$ in $F(n)$ and the logarithmic asymptotics of $F(n)$}\label{section_F_n}

It has been proved above
that every $n$-state 2PerFA can be transformed to an equivalent 1PerFA with
\begin{equation*}
	F(n) = \max_{\substack{k, \ell, m\\ k > 0,\; \ell \geqslant m \geqslant 0\\m \geqslant \ell - k + 1}} G(k, \ell, m)
\end{equation*}
states, where
\begin{equation*}
	G(k, \ell, m) = k \cdot { \ell \choose m} \cdot {k - 1 \choose \ell - m} \cdot (\ell - m)!
\end{equation*}
What is the optimal partition of $n$ states into $k$ right-moving and $\ell$ left-moving states,
and what is the optimal number $m$ of unused states at the left end-marker?
This question is answered in the following lemma.

\begin{lemma}\label{F_maximum_lemma}
For every fixed $n = k + \ell$ the function
\begin{equation*}
	G(k, \ell, m) = k \cdot { \ell \choose m} \cdot {k - 1 \choose \ell - m} \cdot (\ell - m)!
\end{equation*}
is defined for $k > 0$, $\ell \geqslant m \geqslant 0$, $m \geqslant \ell - k + 1$,
and reaches its maximum value on a triple $(k, \ell, m)$, where $k > \ell$.
If $n \geqslant 8$, then the optimal values of the arguments are:
\begin{equation*}
	k = \bigg\lfloor \frac{n + 2}{2} \bigg\rfloor, \quad \ell = \bigg\lceil \frac{n - 2}{2} \bigg\rceil, \quad m = 
	\begin{cases}
		\lceil \frac{\sqrt{3 + 2n} - 3}{2} \rceil, & n \text{ is odd} \\
		\lceil \frac{\sqrt{4 + 2n} - 4}{2} \rceil, & n \text{ is even} \\
	\end{cases}
\end{equation*}
\end{lemma}
\begin{proof}[Sketch of a proof]
	To prove this, firstly, find an optimal value of $m$ for a fixed pair $(k, \ell)$. Denote it by $m_{\text{opt}} = m_{\text{opt}}(k, \ell)$. Then analyse the next ratio 
	\begin{equation*}
		\frac{G(k, \ell, m_{\text{opt}}(k, \ell))}{G(k + 1, \ell - 1, m_{\text{opt}}(k + 1, \ell - 1))}
	\end{equation*}
	It is proved that this ratio is at least $1$ if $k > \ell$ and $n \geqslant 8$, and at most $1$ if $k \leqslant \ell$. Therefore, the optimal partition $n = k + \ell$ has $k > \ell$. If $n \geqslant 8$, then it has $k = \ell + 1$ or $k = \ell + 2$, depending on the parity of $n$, and the optimal values of $k$ and $\ell$ are
	\begin{equation*}
		 k = \bigg\lfloor \frac{n + 2}{2} \bigg\rfloor, \quad \ell = \bigg\lceil \frac{n - 2}{2} \bigg\rceil
	\end{equation*}
	There is a formula for $m_{\text{opt}}(k, \ell)$, and its value for approximately equal $k$ and $\ell$ is
	\begin{equation*}
		m_{\text{opt}}(k, \ell) = \bigg\lceil \frac{\sqrt{D} + \ell - k - 2}{2} \bigg\rceil, \quad \text{where } D = (k - \ell)^2 + 4(\ell + 1)
	\end{equation*}
	Then, for a given $n \geqslant 8$, the claimed optimal value of $m$ can be found by substituting the optimal values of $k$ and $\ell$ into the formula for $m_{\text{opt}}$.
\end{proof}

With the optimal values of $k$, $\ell$ and $m$ determined,
the main result of this paper can now be finally stated.

\begin{theorem}
Let $n \geqslant 1$.
For every $n$-state 2PerFA there exists a 1PerFA with $F(n)$ states
that recognizes the same language,
and in the worst case $F(n)$ states in a 1PerFA are necessary.
\end{theorem}
\begin{proof}
The upper bound is given in Corollary~\ref{upper_bound_corollary}.

For the lower bound, for every $n \geqslant 8$,
let $k$, $\ell$ and $m$ be as in Lemma~\ref{F_maximum_lemma}.
Then, since $k>\ell$,
Theorem~\ref{lower_bound_theorem} presents the desired $n$-state 2PerFA,
for which every 1PerFA recognizing the same language
must have at least $G(k, \ell, m)=F(n)$ states.

For $n=5,6,7$, a calculation of possible values $k, \ell, m$ shows that the maximum of $G(k, \ell, m)$
is reached for $m=1$.
Then, Theorem~\ref{lower_bound_theorem} is still applicable
and provides the witness languages.

For $n=4$, the optimal values given by a calculation
are $k=3$, $\ell=1$ and $m=0$.
Nevertheless, the same automaton as in Theorem~\ref{lower_bound_theorem}
still provides the desired lower bound,
which was checked by a computer calculation.

Finally, $F(n)=n$ for $n=1,2,3$, and (trivial) witness languages are $(a^n)^*$.
\end{proof}

So, $F(n) = G(k, \ell, m)$ for the specified values of $k, \ell, m$. As
\begin{equation*}
	G(k, \ell, m) = k \cdot { \ell \choose m} \cdot {k - 1 \choose \ell - m} \cdot (\ell - m)! = \frac{k! \ell!}{(k - 1 - \ell + m)! m! (\ell - m)!}
\end{equation*}
the asymptotics of $F(n)$ can be determined by
using Stirling's approximation of factorials for $k, \ell, m$ from the optimal partition.
The final result is given in the following theorem.
\begin{theorem}\label{F_n_asymptotics_theorem}
	$F(n) = n^{\frac{n}{2} - \frac{1 + \ln 2}{2}\frac{n}{\ln n}(1 + o(1))}$.
\end{theorem}

To compare, transformation of 2DFA of the general form to 1DFA
has the sharp bound proved by Kapoutsis~\cite{Kapoutsis_thesis}.
\begin{equation*}
	n(n^n - (n - 1)^n) + 1
\end{equation*}
The transformation of sweeping 2DFA to 1DFA~\cite{GeffertOkhotin_2dfa_to_1dfa}
requires slightly fewer states, yet still of the order $n^{n(1+o(1))}$.
\begin{equation*}
	\varphi(n)
		=
	\max\limits_{k=1}^n k^{n-k+1} + 1
		=
	n^{n - \frac{n \ln \ln n}{\ln n} + O(\frac{n}{\ln n})} 
\end{equation*}
Evidently, in the case of 2PerFA,
the cost of transformation to one-way is substantially reduced
(with the exponent divided by two).

The transformation complexity in these three cases
is compared for small values of $n$ in Table~\ref{t:F_vs_Kapoutsis}.

\begin{table}[h]
\begin{center}
\begin{tabular}{r@{\;\;}|@{\;\;}r@{\;\;}|@{\;\;}r@{\;\;}|@{\;\;}r@{\;\;}|}
$n$ & $F(n)$ & $\max\limits_{k=1}^n k^{n-k+1} + 1$ & $n(n^n - (n-1)^n) + 1$ \\
& \textbf{(2PerFA to 1DFA)} & (sweeping to 1DFA) & (2DFA to 1DFA) \\
\hline
1   & 1 & 2 & 2 \\
2   & 2 & 3 & 7 \\
3   & 3 & 5 & 58 \\
4   & 6 & 10 & 701 \\
5   & 12	& 28	& 10506 \\
6   & 24	& 82	& 186187 \\
7   & 72	& 257   & 3805250 \\
8   & 180   & 1025  & 88099321 \\
9   & 480   & 4097  & 2278824850 \\
10  & 1440  & 16385 & 65132155991 \\
11  & 3600  & 78126 & 2038428376722 \\
12  & 12600 & 390626	& 69332064858421 \\
\hline
\end{tabular}
\end{center}
\caption{The value of $F(n)$ compared to the known transformations for irreversible 2DFA, for small values of $n$.}
\label{t:F_vs_Kapoutsis}
\end{table}

\section{A more general definition}\label{section_more_general}

Consider a variant of the definition of a 2PerFA, in which acceptance is also allowed at the left end-marker in states from $Q_-$. It entails that in a transformation from 2PerFA to 1PerFA in each state $(q, f)$ the function $f$ operates from $Q_-$ to $Q_+ \cup \{\text{ACC}, \text{REJ} \}$. Upon reading a string $u \in \Sigma^*$, the 1PerFA comes to a state $(q, f)$, where $q \in Q_+$ is the state in which the 2PerFA first moves to the right from the last symbol of ${\vdash}u$, and for every state $r \in Q_-$ and $p \in Q_+$, if $f(r) = p$, then the 2PerFA after reading ${\vdash}u$ starting at its last symbol in the state $r$ finishes in the state $p$. If $f(r) = \text{REJ}$ then, after reading ${\vdash}u$ from right to left starting in $r$, the 2PerFA rejects at the left end-marker. And if $f(r) = \text{ACC}$, then, after reading ${\vdash}u$ from right to left starting in $r$ the 2PerFA accepts at the left end-marker. The transition function $\delta$, the initial state and the set of accepting states will be defined similarly to Lemma~\ref{2PerFA_to_1PerFA_lemma}. The automaton constructed by this transformation will be a permutation automaton. To prove this claim, $\delta$ is first shown to be injective, and then bijectivity follows from the equality of its domain and range.

As in the proof of Lemma~\ref{2PerFA_to_1PerFA_lemma}, suppose that $\delta$ is not injective, and has the same value on two different states:
\begin{equation*}
	\delta((q_1, f_1)) = \delta((q_2, f_2))
\end{equation*}
\begin{equation*}
	(\delta_a^+(q_1), \delta_a^+ \circ f_1 \circ \delta_a^-) = (\delta_a^+(q_2), \delta_a^+ \circ f_2 \circ \delta_a^-)
\end{equation*}

From
\begin{equation*}
	\delta_a^+(q_1) = \delta_a^+(q_2)
\end{equation*}
follows
\begin{equation*}
	q_1 = q_2
\end{equation*}
as $\delta_a^+$ is a bijection.
And from
\begin{equation*}
	\delta_a^+ \circ f_1 \circ \delta_a^- = \delta_a^+ \circ f_2 \circ \delta_a^-
\end{equation*}
by multiplying by inverse functions of $(\delta_a^+)^{-1}, (\delta_a^-)^{-1}$ from the left side and from the right side respectively, the next equation follows
\begin{equation*}
	(\delta_a^+)^{-1} \circ \delta_a^+ \circ f_1 \circ \delta_a^- \circ (\delta_a^-)^{-1} = (\delta_a^+)^{-1} \circ \delta_a^+ \circ f_2 \circ \delta_a^- \circ (\delta_a^-)^{-1}
\end{equation*}
\begin{equation*}
	f_1 = f_2
\end{equation*}

So, $(q_1, f_1) = (q_2, f_2)$, therefore $\delta$ is a bijection and the constructed automaton is a permutation automaton.

Denote the number of accepting states in $Q_-$ by $e$. The exact number of states in the constructed 1PerFA is given in the following theorem.

\begin{theorem}\label{general_definition_upper_bound_theorem}
For every 2PerFA $\mathcal{A} = (\Sigma, Q_+, Q_-, q_0, \langle \delta_a^+ \rangle_{a \in \Sigma}, \langle \delta_a^- \rangle_{a \in \Sigma}, \delta_{\vdash}, \delta_{\dashv}, F)$ with $F \subseteq Q_+ \cup Q_-$ and
$$|Q_+| = k,\quad |Q_-| = \ell,\quad |Q_-^{\times}| = m,\quad |F \cap Q_-| = e,$$
where $Q_-^{\times} \subseteq Q_-$~is the set of rejecting states from which there is no transition by $\vdash$, there exists a 1PerFA with at most
\begin{equation*}
k \cdot {\ell \choose m} \cdot {m \choose e} \cdot {k - 1 \choose \ell - m} \cdot (\ell - m)!
\end{equation*}
states that recognizes the same language.
\end{theorem}

\section{Conclusion}

The complexity of transforming sweeping permutation automata (2PerFA)
to classical one-way permutation automata (1PerFA) has been determined precisely.
A suggested question for future research is the state complexity of operations on 2PerFA.
Indeed, state complexity of operations on 1PerFA
has recently been investigated~\cite{HospodarMlynarcik,RauchHolzer},
state complexity of operations on 2DFA of the general form
was studied as well~\cite{JiraskovaOkhotin_2dfa},
and it would be interesting to know how the case of 2PerFA compares to these related models.

\end{document}